\documentclass[onecolumn,draftcls,12pt]{IEEEtran}	

\usepackage[noadjust]{cite}

\ifCLASSINFOpdf
  \usepackage[pdftex]{graphicx}
\else
  \usepackage[dvips]{graphicx}
\fi
\graphicspath{{./figs/}}

\usepackage[T1]{fontenc} 
\usepackage{mathtools,amsmath,amssymb,amsthm}
\usepackage[cmintegrals]{newtxmath}
\usepackage{bm} 

\usepackage{algorithmic}

\usepackage{array}

\ifCLASSOPTIONcompsoc
  \usepackage[caption=false,font=normalsize,labelfont=sf,textfont=sf]{subfig}
\else
  \usepackage[caption=false,font=footnotesize]{subfig}
\fi

\usepackage[hyphens]{url}

\usepackage[ruled,linesnumbered,lined]{algorithm2e}
\usepackage{etoolbox}
\makeatletter
\patchcmd{\@algocf@start}
  {-1.5em}
  {0pt}
  {}{}
\makeatother

\newtheorem{theorem}{Theorem}
\DeclareMathOperator*{\maximize}{maximize}
\DeclareMathOperator*{\minimize}{minimize}
\DeclareMathOperator*{\argmax}{argmax}

\DeclareMathOperator*{\subjto}{\textnormal{subject to}}

\newcommand{\suchthat}{%
  \nonscript\;
  \ifnum\currentgrouptype=16
    \middle|
  \else
    \@suchthat
  \fi
  \nonscript\;
}

\begin{document}

\title{Dynamic Joint Scheduling of Anycast Transmission and Modulation in Hybrid Unicast-Multicast SWIPT-Based IoT Sensor Networks}

\author{Do-Yup Kim, \IEEEmembership{Member, IEEE,} Chae-Bong Sohn, and Hyun-Suk Lee
\thanks{D.-Y. Kim is with the Department of Information and Communication AI Engineering, Kyungnam University, Changwon-si, Gyeongsangnam-do 51767, South Korea (e-mail: doyup09@kyungnam.ac.kr).}
\thanks{C.-B. Sohn is with the Department of Electronics and Communications Engineering at Kwangwoon University, Seoul 01897, South Korea (e-mail: cbsohn@kw.ac.kr).}
\thanks{H.-S. Lee is with the Department of Intelligent Mechatronics Engineering, Sejong University, Seoul 05006, South Korea (e-mail: hyunsuk@sejong.ac.kr).}}


\maketitle

\begin{abstract}
\begin{center}\includegraphics[width=5in]{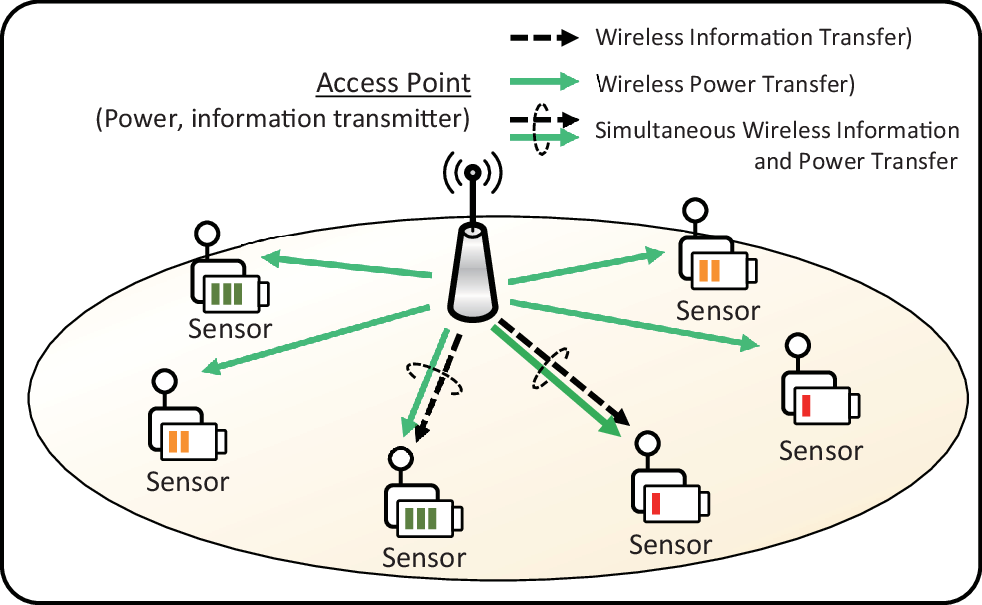}\end{center}
The separate receiver architecture with a time- or power-splitting mode, widely used for simultaneous wireless information and power transfer (SWIPT), has a major drawback: Energy-intensive local oscillators and mixers need to be installed in the information decoding (ID) component to downconvert radio frequency (RF) signals to baseband signals, resulting in high energy consumption.
As a solution to this challenge, an integrated receiver (IR) architecture has been proposed, and, in turn, various SWIPT modulation schemes compatible with the IR architecture have been developed.
However, to the best of our knowledge, no research has been conducted on modulation scheduling in SWIPT-based IoT sensor networks while taking into account the IR architecture.
Accordingly, in this paper, we address this research gap by studying the problem of joint scheduling for unicast/multicast, IoT sensor, and modulation (UMSM) in a time-slotted SWIPT-based IoT sensor network system.
To this end, we leverage mathematical modeling and optimization techniques, such as the Lagrangian duality and stochastic optimization theory, to develop an UMSM scheduling algorithm that maximizes the weighted sum of average unicast service throughput and harvested energy of IoT sensors, while ensuring the minimum average throughput of both multicast and unicast, as well as the minimum average harvested energy of IoT sensors.
Finally, we demonstrate through extensive simulations that our UMSM scheduling algorithm achieves superior energy harvesting (EH) and throughput performance while ensuring the satisfaction of specified constraints well.
\end{abstract}


\section{introduction}

\label{sec:intro}

The internet-of-things (IoT) has revolutionized the way we interact with devices, providing intelligent control and efficient solutions for daily tasks \cite{al2015internet, khan2021federated, nguyen20226g, lee2022radio}.
Among various IoT applications, IoT sensor networks consisting of a number of IoT sensors have gained significant attention in many fields due to their versatility, such as connecting a wide range of IoT devices \cite{mukhopadhyay2021artificial, jamshed2022challenges}.
However, their energy-constrained nature poses a significant challenge in maintaining reliable and long-term operation.
To address this challenge, wireless power transfer (WPT) has emerged as a promising solution, allowing IoT sensors to harvest energy from radio frequency (RF) signals, optical signals, etc. \cite{li2020quasi, wuthibenjaphonchai2021wearable, truong2020fundamental, nabavi2022efficient}.
Simultaneous wireless information and power transfer (SWIPT) takes this a step further by enabling wireless sensors to harvest energy from RF signals and decode information simultaneously, making it a game-changer for IoT sensor networks \cite{prathima2022uav, gao2022swipt, shukla2022exploiting, jia2022joint, nguyen2023physical, chen2021energy, ji2020swipt, rauniyar2020performance}.

The original concept of SWIPT was first proposed in \cite{varshney2008transporting}, where the fundamental tradeoff between transmitting energy and transmitting information simultaneously over a single noisy line was investigated.
Since then, researchers from various fields have conducted numerous studies to enhance various performance metrics of SWIPT by considering time-splitting and power-splitting modes \cite{zhang2013mimo, liu2013wireless, liu2013wireless2, shi2014joint}.
In these modes, the receiver architecture typically consists of separate components for information decoding (ID) and energy harvesting (EH). 
In the time-splitting mode, wireless information transfer (WIT) and WPT are performed alternately in time by the ID and EH components, respectively.
In the power-splitting mode, incoming RF signals are split in the power domain and conveyed exclusively to the ID and EH components, and WIT and WPT are carried out independently.
This receiver architecture offers a significant advantage as conventional WIT technologies in conjunction with traditional RF receivers can be easily applied and leveraged, with the simple additional consideration that the received signal periodically enters the ID component in the time-splitting mode, and the received signal with reduced power is input to the ID component in the power-splitting mode.
Despite this strong ease-of-use advantage, this receiver architecture has a critical limitation, which is that local oscillators and mixers need to be installed in the ID component to downconvert RF to baseband signals, resulting in high energy consumption.
Hence, this receiver architecture may be impractical for energy-constrained IoT sensor networks composed primarily of (ultra-)low-power IoT sensors \cite{yang2017hardware, shafiee2017infrastructure, jang2017circuit, jang2017circuit2}.

To realize SWIPT for low-power IoT sensors, a different type of receiver architecture has recently been proposed, referred to as an integrated receiver (IR) architecture \cite{zhou2013wireless}.
In this architecture, the received RF signal is first rectified and converted into a direct current (DC) signal by a rectifier circuit.
Then, the rectified DC signal is fully exploited for EH, while its strength is detected and utilized for ID.
In other words, unlike the conventional receiver architecture, the rectifier circuit in the IR architecture is designed to serve a dual purpose.
Composed simply of diodes and a few passive elements like capacitors and resistors, it is used both for EH and for ID.
As a result, the IR architecture can fully exploit the received signal for EH, improving the SWIPT system's power delivery performance.
However, compared to traditional RF receivers, since the rectifier erases information contained in phases of RF signals, modulation schemes are limited to amplitude-based ones that utilize the rectified signal's amplitudes for information transfer, resulting in somewhat deteriorated information delivery performance.
Nevertheless, the IR architecture greatly reduces power consumption in the ID process by eliminating energy-consuming RF components, including local oscillators and mixers, used in traditional RF receivers, making it a promising solution for low-power IoT sensors.

To take advantage of the aforementioned benefits, several amplitude-based modulation schemes compatible with an IR architecture have been proposed \cite{zhou2013wireless, kim2016new, claessens2018enhanced, krikidis2019tone, rajabi2018modulation, claessens2019multitone, im2020multi, kim2022wireless}.
In \cite{zhou2013wireless}, the authors proposed a pulse energy modulation (PEM) scheme that conveys information through the amplitude of the rectified DC signal.
Mathematically, PEM can be considered equivalent to amplitude shift keying (ASK) under the assumption that all constellation points are posed only in the positive region \cite{tse2005fundamentals}.
With the revelation that multisine waveforms with high peak-to-average-power-ratios (PAPRs) lead to high EH performance \cite{clerckx2016waveform}, the authors in \cite{kim2016new} proposed a new modulation scheme that exploits the PAPR value of the rectified DC signal in determining which symbol is transmitted in the SWIPT system using multisine waveforms as transmitting RF signals.
The authors additionally analyzed the bit error rate (BER) and the rate-energy tradeoff.
In \cite{claessens2018enhanced}, the PEM was extended to biased-ASK (BASK), whose effective constellation region starts from a certain positive value other than zero, thereby enhancing the EH performance significantly at the expense of the BER performance.
In \cite{krikidis2019tone}, the authors proposed the tone index multisine modulation scheme, which conveys information through the number of tones used.
Then, in \cite{im2020multi}, the multi-tone amplitude modulation (MAM) scheme that conveys information through the combination of the amplitude and the number of subcarriers used was proposed.
MAM can be considered the generalized version of the modulation schemes proposed in \cite{zhou2013wireless, kim2016new, claessens2018enhanced, krikidis2019tone}.
In \cite{rajabi2018modulation} and \cite{claessens2019multitone}, the authors proposed modulation schemes using RF signals consisting of multiple tones with different frequencies.
In these schemes, information is conveyed through the magnitudes of the rectified signal's DC component and intermodulation product component, or the input frequency spacings and frequencies of the rectifier output intermodulation product, respectively.
In \cite{kim2022wireless}, the authors proposed a pulse position modulation scheme in which information is encoded in the position of the pulse.

Although various modulation schemes for SWIPT compatible with an IR architecture have been proposed, all the works mentioned above have considered a simple receiver architecture with a circuit consisting of one diode and one resistor–capacitor (RC) filter, thereby limiting the emergence of new ideas structurally.
In contrast, in our previous work \cite{kim2019dual}, we considered an IR architecture for SWIPT that leverages a double half-wave rectifier circuit, consisting of two pairs of diodes and capacitors.
This circuit rectifies the received signal in different directions, yielding two different DC signals: a positive DC signal and a negative one.
Accordingly, information can be encoded based on any arbitrary function of them.
We named this modulation methodology dual amplitude shift keying (DASK).
Furthermore, we present two realizations of DASK using two different encoding functions, namely, amplitude ratio shift keying (ARSK) and amplitude difference shift keying (ADSK).
It has already been demonstrated that these modulation schemes achieve not only higher BER performance but also higher EH performance compared to the existing modulation schemes, such as PEM and BASK.

Many studies, including \cite{zhou2013wireless, kim2016new, claessens2018enhanced, krikidis2019tone, rajabi2018modulation, claessens2019multitone, im2020multi, kim2022wireless} and our own \cite{kim2019dual}, have investigated SWIPT modulations compatible with an IR architecture.
However, these studies have primarily focused on peer-to-peer communications and have not explored the utilization of SWIPT modulations.
To the best of our knowledge, no research has investigated modulation scheduling in SWIPT-based IoT sensor networks consisting of multiple IoT sensors with IR architectures.
Therefore, in this paper, we aim to fill this research gap by utilizing ARSK and ADSK in a time-slotted SWIPT-based IoT sensor network system where a hybrid access point (H-AP) delivers data streams and energy via RF signals to multiple IoT sensors participating in several tasks.
Additionally, the appropriate use of multicast transmission, which delivers the same data stream or software updates simultaneously to multiple IoT sensors participating in a specific task, can be particularly effective in SWIPT-based IoT sensor networks  \cite{hao2019energy, gautam2020multigroup, mishra2018energy, tan2021energy}, compared to unicast transmission, which sends a separate data stream to each individual IoT sensor.
Hence, we address the problem of joint scheduling for unicast/multicast, IoT sensor, and modulation (UMSM) in the time-slotted SWIPT-based IoT sensor network.
Specifically, the objective is to optimize scheduling to maximize the weighted sum of the average unicast service throughput and harvested energy of IoT sensors while ensuring the minimum average throughput of both multicast and unicast, as well as the minimum average harvested energy of IoT sensors.
To achieve this goal, we develop a UMSM scheduling algorithm that leverages the stochastic nature of wireless fading channels to dynamically determine the appropriate communication service type (multicast or unicast), select the most suitable task (if multicast) or sensor (if unicast), and choose the optimal modulation scheme between ARSK and ADSK for each time slot.
Our approach combines mathematical modeling and optimization techniques, such as the Lagrangian duality and stochastic optimization theory \cite{park2022joint, kim2022low}.
To the best of our knowledge, this is the first work that jointly considers communication service type scheduling, as well as modulation scheduling, in SWIPT-based IoT sensor networks, taking into account the IR architecture.
Finally, we evaluate the performance of our proposed algorithms through extensive simulations and demonstrate the superiority of our algorithm in terms of both EH and throughput performance.

The remaining sections of the paper are structured as follows.
In Section \ref{sec:sys}, we present the system model for SWIPT-based IoT sensor networks.
In Section \ref{sec:scheduling}, we formulate the UMSM scheduling problem and present an algorithm that solves it, called the UMSM scheduling algorithm.
In Section \ref{sec:joint_selection}, we develop the UMSM selection algorithm, an internal algorithm that runs every time slot within the UMSM scheduling algorithm.
The simulation results are presented in Section \ref{sec:sim}, and our conclusions and recommendations based on the results are provided in Section \ref{sec:conc}.

\section{System Model}
\label{sec:sys}

\subsection{Signal Model, Receiver Model, and Modulation Scheme}

\begin{figure}[!t]
\centering
\subfloat[Architecture model.]{\includegraphics[width=.6\columnwidth]{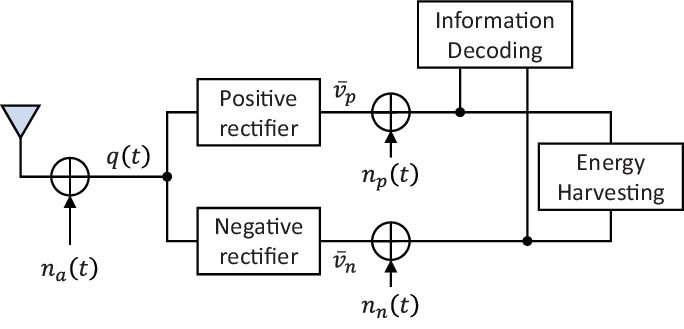}\label{fig:Receiver architecture}}
\hfil
\subfloat[Circuit model.]{\includegraphics[width=.6\columnwidth]{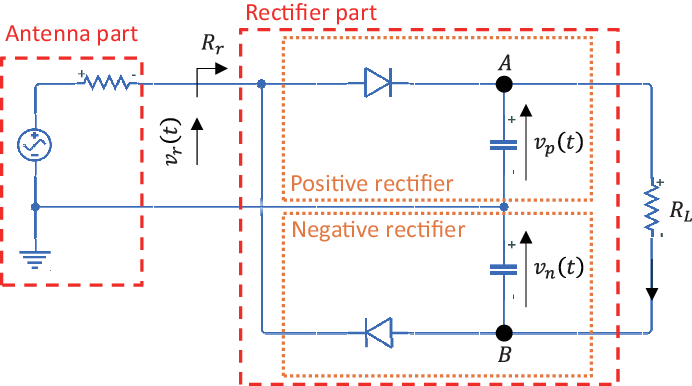}\label{fig:Receiver circuit}}
\caption{A proposed IR model for SWIPT.}
\label{fig:Receiver}
\end{figure}

In our previous work \cite{kim2019dual}, we have proposed an IR architecture for SWIPT, as shown in Fig. \ref{fig:Receiver}, and introduced a novel modulation methodology called DASK.
We have presented two representative examples of this methodology, namely ARSK and ADSK, and demonstrated their superior performance compared to conventional schemes, such as PEM and BASK \cite{zhou2013wireless, claessens2018enhanced}.
In this subsection, we tailor the findings from \cite{kim2019dual} to the system model that we consider in this paper and provide the relevant background information to exploit the benefits of using ARSK and ADSK.
Specifically, we describe the signal and receiver models, and define the performance metrics for EH and ID.

In many literature related to WPT and SWIPT (e.g., \cite{kim2016new, krikidis2019tone, rajabi2018modulation, claessens2019multitone, im2020multi, kim2022wireless, clerckx2016waveform}), RF signals of multisine waveforms, also known as tones, have been widely used due to their high EH efficiency.
Hence, although not necessarily, we also adopt general RF signals composed of $N$ sine waves, with the index set denoted as $\mathcal{F}=\{1,2,\ldots,N\}$.
Additionally, we consider an $M$-ary transmission based on a symbol set with $M$ symbols, with the index set denoted as $\mathcal{M}=\{1,2,\ldots,M\}$.
Then, the transmitted signal for symbol $m\in\mathcal{M}$ can be expressed as
\begin{equation}
	s_m(t) = \sum_{n\in\mathcal{F}} a_{n,m} \cos (2 \pi f_n t + \phi_{n,m} ),
\end{equation}
where $a_{n,m}$ and $\phi_{n,m}$ are the amplitude and phase of the $n$th tone with frequency $f_n$ for the $m$th symbol, respectively.
We consider the transmitted power constraint, due to the transmitter's power budget, as
\begin{equation}
	\mathbb{E}\{\lvert s_m(t)\rvert^2\} \le P_t,
\end{equation}
where $\mathbb{E}$ is the expectation operator, and $P_t$ is the transmit power budget of the transmitter.
Accordingly, the received signal is given by
\begin{equation} \label{eq:rm}
	q_m(t) = \sum_{n\in\mathcal{F}} h_n a_{n,m} \cos (2 \pi f_n t + \phi_{n,m} ) + n_a(t),
\end{equation}
where $n_a(t)$ is the antenna noise, and $h_n$ is the wireless channel gain for the $n$th tone. 
Note that the wireless channel gain will be modeled in detail later in Section \ref{sec:swipt_system_model}.

Fig. \ref{fig:Receiver} illustrates our proposed IR architecture for SWIPT, which consists of an architecture model shown in Fig. \ref{fig:Receiver architecture} and a circuit model in Fig. \ref{fig:Receiver circuit}.
In these figures, $q(t)$ represents the received signal obtained from \eqref{eq:rm},\footnote{For brevity, unless there is confusion, we omit the subscript $m$ in $q_m(t)$.} $\bar{v}_p$ and $\bar{v}_n$ are the amplitudes of the output DC signals rectified in different directions by the positive and negative rectifiers, respectively, $n_p(t)$ and $n_n(t)$ are the additive rectifier noises modeled as normal distributions with zero mean and variance $\sigma_p^2$ and $\sigma_n^2$, respectively, $R_r$ is the input impedance of the rectifier part, and $R_L$ is the load resistance corresponding to the EH module.
It is worth noting that Fig. \ref{fig:Receiver architecture} and Fig. \ref{fig:Receiver circuit} are related to each other, as $v_r(t) = q(t)\sqrt{R_r}$ and the voltage potentials at points $A$ and $B$ in the circuit model correspond to $\bar{v}_p$ and $\bar{v}_n$, respectively, in the architecture model.

Through the circuit analysis based on the Shockley diode nonlinear equation $i_d = I_s (\exp(v_d/(\rho V_{th}))-1)$, where $i_d$ is the diode current, $I_s$ is the reverse bias saturation current, $v_d$ is the voltage across the diode, $\rho$ is the ideality factor, and $V_{th}$ is the thermal voltage, we can derive closed-form expressions for $\bar{v}_p$ and $\bar{v}_n$, respectively, as
\begin{align}
	\bar{v}_p &=  \rho V_{th} \left(\mathcal{W} \left(\chi\sqrt{g_1g_2} \exp(\chi)/T \right) - \chi + \ln \sqrt{g1/g2}\right), \label{eq:vp}\\
	\bar{v}_n &= -\rho V_{th} \left(\mathcal{W} \left(\chi\sqrt{g_1g_2} \exp(\chi)/T \right) - \chi - \ln \sqrt{g1/g2}\right), \label{eq:vn}
\end{align}
where $\mathcal{W}(\cdot)$ is the Lambert W-function \cite{corless1996lambert}, $T$ is the period of the input signal to the rectifier, $\chi=(I_sR_L)/(2\rho V_{th})$, $g_1=\int_T \exp(v_r(t)/(\rho V_{th})) \, dt$, and $g_2=\int_T \exp(-v_r(t)/(\rho V_{th})) \, dt$.
Without loss of generality, we assume that $\bar{v}_p$ is positive, and accordingly, $\bar{v}_n$ is negative.
The detailed derivation process of \eqref{eq:vp} and \eqref{eq:vn} is given in \cite{kim2019dual}.
Then, the average harvested energy per unit time, denoted by $\varepsilon$, can be given as
\begin{equation}
	\varepsilon = \mathbb{E} \left\{ \frac{(\bar{v}_{p,m}-\bar{v}_{n,m})^2}{R_L} \right\} = \frac{1}{\lvert\mathcal{M}\rvert} \sum_{m\in\mathcal{M}} \frac{(\bar{v}_{p,m}-\bar{v}_{n,m})^2}{R_L}, \label{eq:E}
\end{equation}
where $\lvert\cdot\rvert$ for a set denotes the cardinality operator, and $\bar{v}_{p,m}$ and $\bar{v}_{n,m}$ correspond to $\bar{v}_p$ and $\bar{v}_n$ for symbol $m$, respectively.
That is, they are the two amplitudes of the output DC signal with respect to $q_m(t)$.

Now, we explain the modulation methodology of DASK.
The main idea, which differs from conventional methods, is to convey information through a combination of $\bar{v}_p$ and $\bar{v}_n$, i.e., the value of $\psi(\bar{v}_p, \bar{v}_n)$, where $\psi$ is a symbol mapping function that can be defined arbitrarily as long as it is a one-to-one mapping.
Therefore, depending on how we set up $\psi$, DASK can be implemented in various ways.
Among these various implementations, ARSK and ADSK have been presented in \cite{kim2019dual}, and it has already been demonstrated that they have superior performance compared to existing methods.
Hence, we adopt ARSK and ADSK as candidate modulation schemes in this paper.
For simple notation, we hereafter use superscripts $(\cdot)^R$ and $(\cdot)^D$ to represent ARSK and ADSK, respectively.
Then, the symbol mapping functions for ARSK and ADSK are, respectively, defined by
\begin{align}
	\psi^R (\bar{v}_p, \bar{v}_n) &= \lvert \bar{v}_p \rvert / \lvert \bar{v}_n \rvert, \label{eq:phiR}\\
	\psi^D (\bar{v}_p, \bar{v}_n) &= \bar{v}_p - \bar{v}_n = \lvert \bar{v}_p \rvert + \lvert \bar{v}_n \rvert. \label{eq:phiD}
\end{align}
With the widely accepted assumption that antenna noise is dominated by rectifier noise \cite{zhang2013mimo}, the transmitted symbol can be determined based on the observed value of $\psi(\bar{v}_p+n_p, \bar{v}_n+n_n)$, where $n_p\sim\mathcal{N}(0,\sigma_p^2)$ and $n_n\sim\mathcal{N}(0,\sigma_n^2)$, 
According to \eqref{eq:phiR} and \eqref{eq:phiD}, the observed values for ARSK and ADSK follow random variables $\Omega^R = \Omega_p/\Omega_n$, where $\Omega_p\sim\mathcal{N}(\bar{v}_p,\sigma_p^2)$ and $\Omega_n\sim\mathcal{N}(\bar{v}_n,\sigma_p^2)$, and $\Omega^D\sim\mathcal{N}(\bar{v}_p-\bar{v}_n, \sigma_p^2+\sigma_n^2)$, respectively.
Accordingly, their conditional probability density functions (PDFs) when the $m$th symbol has been transmitted can be given in closed form, respectively, as
\begin{align}
	f_{\Omega^R} (\omega \mid m) &= \frac{\beta(\omega)\cdot \theta(\omega)}{\sqrt{2\pi} \sigma_p \sigma_n \alpha^3(\omega)} \left[ Q\left(-\frac{\beta(\omega)}{\alpha(\omega)} \right) - Q\left(\frac{\beta(\omega)}{\alpha(\omega)}\right)\right] \nonumber\\
	&\quad +\frac{\exp(-\gamma/2)}{\pi\sigma_p\sigma_n \alpha^2(\omega)} \label{eq:pdfR}\\
	f_{\Omega^D} (\omega \mid m) &= \frac{1}{\sqrt{2\pi(\sigma_p^2+\sigma_n^2)}} \exp\left(-\frac{(\omega-\bar{v}_{p,m}+\bar{v}_{n,m})^2}{2(\sigma_p^2+\sigma_n^2)}\right), \label{eq:pdfD}
\end{align}
where $\alpha(\omega)=\sqrt{\frac{\omega^2}{\sigma_p^2} + \frac{1}{\sigma_n^2}}$, $\beta(\omega)=\frac{\bar{v}_{p,m}\omega}{\sigma_p^2}-\frac{\bar{v}_{n,m}}{\sigma_n^2}$, $\gamma=\frac{\bar{v}_{p,m}^2}{\sigma_p^2}+\frac{\bar{v}_{n,m}^2}{\sigma_n^2}$, $\theta(\omega)=\exp(\frac{\beta^2(\omega)-\gamma a^2(\omega)}{2a^2(\omega)})$, and $Q(\cdot)$ is the Q-function defined by $Q(x)=\frac{1}{\sqrt{2\pi}} \int_x^\infty e^{-t^2/2}\,dt$.
With the equiprobable symbols, the receiver can optimally estimate the transmitted symbol as $m^* = \argmax_m f_\Omega(\omega\mid m)$, where $\Omega$ is $\Omega^R$ for ARSK and $\Omega^D$ for ADSK.
Based on \eqref{eq:pdfR} and \eqref{eq:pdfD}, the bit error rate (BER) curve with respect to the received power, denoted by $P_r$, can be obtained easily using the Monte Carlo method.
Note that $P_r = \mathbb{E} \{ \lvert q_m(t) \rvert^2 \}$.
Then, by using any curve fitting techniques, we can generate functions $\iota^R$ and $\iota^D$ that map $P_r$ to the BER of ARSK and that of ADSK, respectively.
Thereby, the throughput, measured in bits per channel use (bpcu), can be defined, as in \cite{kim2016new, kim2022wireless}, by
\begin{equation}
	r^i = \left(1-\iota^i(P_r)\right) \frac{\log_2 M}{1+\kappa}, ~ i\in\{R, D\}, \label{eq:Ri}
\end{equation}
where $\kappa$ is the roll-off factor, which is simply assumed to be zero for the minimum Nyquist bandwidth \cite{kim2016new}.
Fig. \ref{fig:Throughput_curves} illustrates a comparison of the throughput curves for ARSK and ADSK as a function of received power when the modulation order is set to $4$, i.e., $M=4$.
In the figure, the true values are obtained based on \eqref{eq:Ri} and Monte Carlo simulation, while the curve-fitting curves are generated using a trained neural network with $10$ hidden layers and the Levenberg-Marquardt algorithm as the training algorithm.
Similarly, Fig. \ref{fig:Energy_curves} compares the harvested energy curves for ARSK and ADSK, where the true values are obtained based on \eqref{eq:E}, and the curve-fitting curves are obtained using a trained neural network with the same structure as in Fig. \ref{fig:Throughput_curves}.
Consequently, if only the received power is determined, the throughput and harvested energy can be achieved by inputting it to the corresponding trained neural networks.
These throughput and harvested energy are regarded as the performance metrics in terms of ID and EH, respectively.

\subsection{SWIPT-Based IoT Sensor Networks}
\label{sec:swipt_system_model}

Consider a time-slotted SWIPT-based IoT sensor network system, where an H-AP serves $K$ IoT sensors\footnote{For the sake of brevity, we will also refer to ``IoT sensors'' simply as ``sensors'' hereafter.} participating in $C$ tasks.
We assume that both the H-AP and sensors are equipped with a single antenna, and the sensors are composed of the IR architecture depicted in Fig. \ref{fig:Receiver}.
We also assume that $C\le K$ to account for scenarios where more than one sensor is participating in the same task.
The index sets for sensors and tasks are denoted by $\mathcal{K}=\{1,\ldots,K\}$ and $\mathcal{C}=\{1,\ldots,C\}$, respectively.
Let $\mathcal{K}(c) \subseteq \mathcal{K}$ represent the set of sensors involved in task $c$, and let $\ell(k)\in\mathcal{C}$ denote the index of the task in which sensor $k$ participates.
In this system, two types of data streams are considered: multicast streams for tasks and unicast streams for individual sensors.
Accordingly, the H-AP determines whether to provide a multicast service for a specific task $c$ to sensors in $\mathcal{K}(c)$ or a unicast service to a particular sensor $k$, based on the given information at the beginning of each time slot.

\begin{figure}[!t]
\centerline{\includegraphics[width=.6\columnwidth]{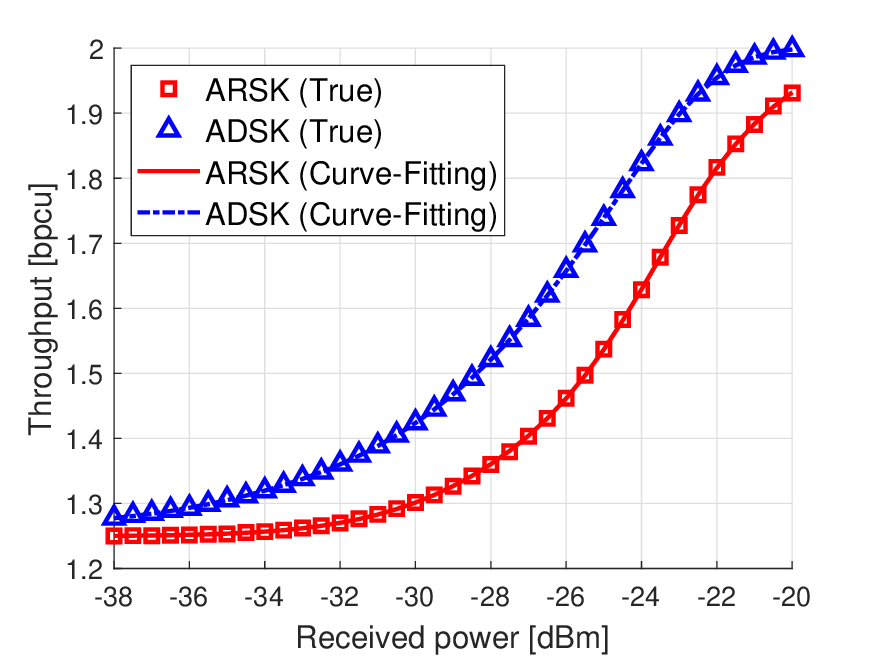}}
\caption{Throughput with respect to the received power.}
\label{fig:Throughput_curves}
\end{figure}
\begin{figure}[!t]
\centerline{\includegraphics[width=.6\columnwidth]{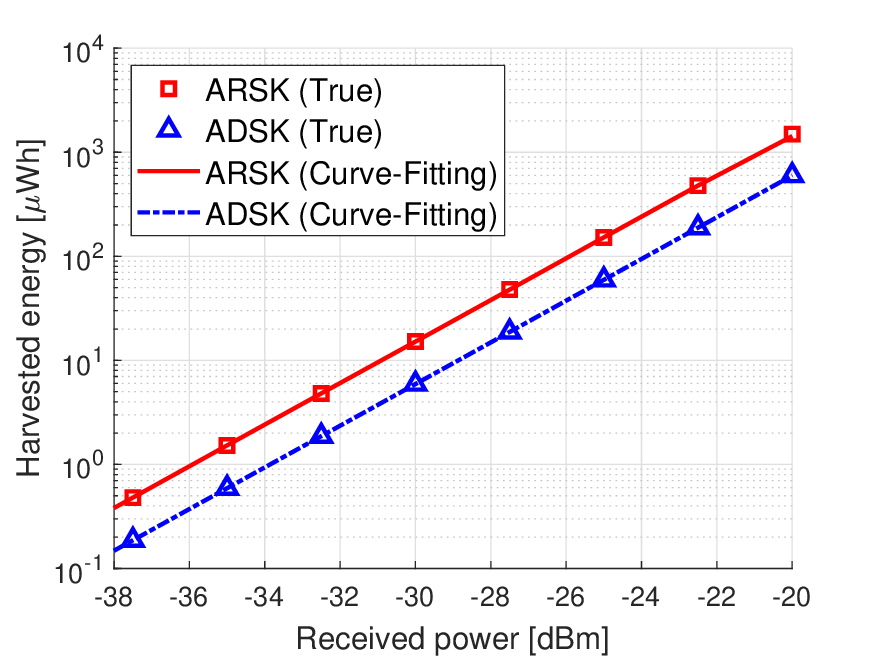}}
\caption{Harvested energy with respect to the received power.}
\label{fig:Energy_curves}
\end{figure}

For the wireless channel model, we assume a block fading channel model, where the channel gain between the H-AP and each sensor is time-varying across time slots and frequency-selective across different tones, but remains constant during a time slot within the same tone.
We denote the channel gain that captures the joint effects of path loss, shadowing, and multipath fading between the H-AP and sensor~$k$ at the $n$th tone in time slot $\tau$ by $h_{k,n}^\tau$.
Note that $h_{k,n}^\tau$ corresponds to $h_n$ in \eqref{eq:rm} for sensor $k$ in time slot $\tau$.
Therefore, the fading process associated with sensor $k$ can be defined by $\{\mathbf{h}_k^\tau, \tau = 1, 2, \ldots\}$, where $\mathbf{h}_k^\tau=\{h_{k,n}^\tau\}_{n\in\mathcal{F}}$ is the channel vector for sensor $k$ in time slot $\tau$.
For convenience, we also define the channel vector of all sensors in time slot $\tau$ as $\mathbf{h}^\tau = \{\mathbf{h}_k^\tau\}_{k\in\mathcal{K}}$.
We assume that the fading process is stationary and ergodic.
To consider a more practical situation, we assume that the H-AP does not have information on the underlying distributions of the fading processes a priori.
However, the instantaneous channel gains are known to the H-AP at the beginning of each time slot through pilot-aided channel estimation techniques.
Based on these channel gains, the H-AP finds a solution to the scheduling problem and makes decisions every time slot accordingly.
We will described these decisions in detail later in Section \ref{sec:scheduling}.

In each time slot, the H-AP determines whether to provide a multicast service for a specific task or a unicast service to a specific sensor, and then, it determines whether to employ ARSK or ADSK as the modulation technique in that time slot.
Accordingly, we define three types of scheduling indicators as                                                                                                                                             
\begin{equation}
	x_c^\tau = \begin{dcases}
		1, & \text{if content for task $c$ is selected to be multicast}\\
		   & \text{to sensors in $\mathcal{K}(c)$ in time slot $\tau$},\\
		0, & \text{otherwise,}
	\end{dcases}
\end{equation}
\begin{equation}
	y_k^\tau = \begin{dcases}
		1, & \text{if sensor $k$ is selected to receive its}\\
		   & \text{unicast service in time slot $\tau$},\\
		0, & \text{otherwise,}
	\end{dcases}
\end{equation}
\begin{equation}
	z^\tau = \begin{dcases}
		1, & \text{if ARSK is employed in time slot $\tau$},\\
		0, & \text{if ADSK is employed in time slot $\tau$.}
	\end{dcases}
\end{equation}
For brevity, let $\mathbf{x}^\tau = \{x_c^\tau\}_{c\in\mathcal{C}}$ and $\mathbf{y}^\tau = \{x_k^\tau\}_{k\in\mathcal{K}}$.
Since at most one stream, either multicast or unicast, can be serviced in each time slot, the scheduling constraints are considered as
\begin{equation}
	\sum_{c\in\mathcal{C}} x_c^\tau + \sum_{k\in\mathcal{K}} y_k^\tau \le 1, ~ \tau = 1, 2, \ldots. \label{eq:constraint_scheduling}
\end{equation}

In the case where sensor $k$ is selected as an information receiver, its effective throughput in time slot $\tau$ is given by
\begin{equation}
	r_k(z^\tau) = z^\tau \cdot r_k^{R,\tau} + (1-z^\tau) \cdot r_k^{D,\tau}, \label{eq:Rktau}
\end{equation}
where $r_k^{i,\tau}$, for $i\in\{R, D\}$, corresponds to $r^i$ in \eqref{eq:Ri} for sensor $k$ in time slot $\tau$ with a channel gain of $\mathbf{h}^\tau$.
For brevity, we use $r_k(z^\tau)$ and $r_k^\tau$ interchangeably unless there is confusion.
Then, the average throughput of sensor $k$ for the multicast service for its involved task, $\ell(k)$, in time slot $\tau$ is given by
\begin{equation}
	B_k = \lim_{T\to\infty} \frac{1}{T} \sum_{\tau=1}^T x_{\ell(k)}^\tau r_k^\tau,
\end{equation}
and, similarly, its average throughput for the unicast service in time slot $\tau$ is given by
\begin{equation}
	U_k = \lim_{T\to\infty} \frac{1}{T} \sum_{\tau=1}^T y_{k}^\tau r_k^\tau.
\end{equation}
We then consider the average throughput requirements for the multicast and unicast services of sensors, respectively, as
\begin{align}
	B_k &\ge \bar{B}_{\ell(k)}, ~ \forall k\in\mathcal{K} \label{eq:constraint_br}\\
	U_k &\ge \bar{U}_{k}, ~ \forall k\in\mathcal{K}, \label{eq:constraint_uni}
\end{align}
where $\bar{B}_c$ is the minimum average throughput of the multicast service for task $c$, and $\bar{U}_k$ is that of the unicast service for sensor $k$.
To simplify the notation, we refer to these constraints as the multicast and unicast constraints, respectively.

Similarly to the formulated throughput above, the effective harvested energy of sensor $k$ in time slot $\tau$ is given by
\begin{equation}
	\varepsilon_k(z^\tau) = z^\tau \cdot \varepsilon_k^{R,\tau} + (1-z^\tau) \cdot \varepsilon_k^{D,\tau}, \label{eq:Ektau}
\end{equation}
where $\varepsilon_k^{R,\tau}$ and $\varepsilon_k^{D,\tau}$ correspond to the harvested energy $\varepsilon$, defined in \eqref{eq:E}, of sensor $k$ in time slot $\tau$ with a channel gain of $\mathbf{h}_k^\tau$ when the transmission signal is generated based on ARSK and ADSK, respectively.
For brevity, we use $\varepsilon_k(z^\tau)$ and $\varepsilon_k^\tau$ interchangeably.
Then, the average harvested energy of sensor $k$ in time slot $\tau$ is defined by
\begin{equation}
	E_k = \lim_{T\to\infty} \frac{1}{T} \sum_{\tau=1}^T \varepsilon_k^\tau.
\end{equation}
We impose the energy constraints by setting requirements for the average harvested energy of sensors, which are given by
\begin{equation}
	E_k \ge \bar{E}_k, ~ \forall k\in\mathcal{K}, \label{eq:constraint_E}
\end{equation}
where $\bar{E}_k$ is the minimum average harvested energy of sensor $k$.
It is noteworthy that unlike the case of receiving information, all sensors, regardless of the scheduling result, harvest energy from the signal transmitted by the H-AP every time slot.
As a result, in each time slot, aside from uncontrollable system conditions such as channel gains, the only impact on EH is which modulation scheme is employed in that time slot.

\section{Scheduling on Unicast/Multicast, IoT Sensor, and Modulation}
\label{sec:scheduling}

In this section, we begin by formulating the UMSM scheduling problem in the SWIPT-based IoT sensor network system.
We then develop the UMSM scheduling algorithm, which determines whether to use multicast and unicast services, selects an appropriate task to serve (if multicast) or sensor (if unicast), and chooses the optimal modulation scheme between ARSK and ADSK for each time slot, thereby solving the UMSM scheduling problem that will be formulated below.

The objective of the UMSM scheduling problem is to determine the optimal scheduling indicators ($\mathbf{x}^\tau$ for multicast services, $\mathbf{y}^\tau$ for unicast services, and $z^\tau$ for modulation schemes) for each time slot $\tau$ in order to maximize the weighted sum of the average unicast service throughput and the average harvested energy of sensors.
This must be achieved while ensuring the minimum average multicast throughput, $\bar{B}_{\ell(k)}$, the minimum average unicast throughput, $\bar{U}_k$, and the minimum average harvested energy, $\bar{E}_k$, of each sensor $k\in\mathcal{K}$.
To formally address this problem, we present the following optimization formulation:
\begin{IEEEeqnarray}{c'l}
	\IEEEyesnumber\label{prob:primal_in_time}
	\maximize_{\mathbf{x}^\tau,\,\mathbf{y}^\tau,\,z^\tau,\,\forall \tau} & \sum_{k\in\mathcal{K}} \left( w_k^u U_k + w_k^\varepsilon E_k \right) \IEEEyessubnumber \\
	\subjto
	 & \eqref{eq:constraint_scheduling}, \eqref{eq:constraint_br}, \eqref{eq:constraint_uni}, \eqref{eq:constraint_E}, \IEEEyessubnumber\\
	& \mathbf{x}^\tau \in \mathcal{X}, ~ \mathbf{y}^\tau \in \mathcal{Y}, ~ z^\tau \in \mathcal{Z}, ~ \forall \tau, \IEEEyessubnumber
\end{IEEEeqnarray}
where $w_k^u$ and $w_k^\varepsilon$ are the weights for the unicast throughput and the harvested energy of sensor $k$, respectively, $\mathcal{X} = \{0,1\}^{\lvert\mathcal{C}\rvert}$, $\mathcal{Y} = \{0,1\}^{\lvert\mathcal{K}\rvert}$, and $\mathcal{Z} = \{0,1\}$.
We simply call $w_k^u$ and $w_k^\varepsilon$ the throughput weight and energy weight of sensor $k$, respectively.
In the problem, dealing with averaging over an infinite time horizon typically renders a problem intractable.
To overcome this, we employ the well-known fact that ergodicity causes the long-term time average to converge with Lebesgue measure $1$ to the expectation for almost all realizations of the fading process.
By using a channel vector, $\mathbf{h}$, in a generic time slot instead of $\mathbf{h}^\tau$ in a specific time slot $\tau$, we can reformulate problem \eqref{prob:primal_in_time} as
\begin{IEEEeqnarray}{c'l}
	\IEEEyesnumber\label{prob:primal_in_h}
	\maximize_{\mathbf{x}^\mathbf{h},\,\mathbf{y}^\mathbf{h},\,z^\mathbf{h},\,\forall\mathbf{h}} & \mathbb{E}_\mathbf{h} \left[ \sum_{k\in\mathcal{K}} \left( w_k^u y_k^\mathbf{h} \cdot r_k(z^\mathbf{h}) + w_k^\varepsilon \cdot \varepsilon_k(z^\mathbf{h}) \right) \right] \IEEEeqnarraynumspace \IEEEyessubnumber \\
	\subjto
	 & \mathbb{E}_\mathbf{h} \left[ x_{\ell(k)}^\mathbf{h} \cdot r_k(z^\mathbf{h}) \right] \ge \bar{B}_{\ell(k)}, ~ \forall k\in\mathcal{K}, \IEEEyessubnumber \label{eq:constraint_br2} \\
	& \mathbb{E}_\mathbf{h} \left[ y_k^\mathbf{h} \cdot r_k(z^\mathbf{h}) \right] \ge \bar{U}_k, ~ \forall k\in\mathcal{K}, \IEEEyessubnumber \label{eq:constraint_uni2} \\
	& \mathbb{E}_\mathbf{h} \left[ \varepsilon_k(z^\mathbf{h}) \right] \ge \bar{E}_k, ~ \forall k\in\mathcal{K}, \IEEEyessubnumber \label{eq:constraint_E2} \\
	& 	\sum_{c\in\mathcal{C}} x_c^\mathbf{h} + \sum_{k\in\mathcal{K}} y_k^\mathbf{h} \le 1, ~ \forall \mathbf{h}, \IEEEyessubnumber \label{eq:constraint_scheduling2} \\
	& \mathbf{x}^\mathbf{h} \in \mathcal{X}, ~ \mathbf{y}^\mathbf{h} \in \mathcal{Y}, ~ z^\mathbf{h} \in \mathcal{Z}, ~ \forall \mathbf{h}, \IEEEyessubnumber \label{eq:xyz_region2}
\end{IEEEeqnarray}
where $\mathbf{x}^\mathbf{h}=\{x_c^\mathbf{h}\}_{c\in\mathcal{C}}$, $\mathbf{y}^\mathbf{h}=\{y_k^\mathbf{h}\}_{k\in\mathcal{K}}$, and $\mathbb{E}_\mathbf{h}(\cdot)$ denotes the expectation over $\mathbf{h}$.
It is worth noting that, in any time slot $\tau$ with a realized channel vector $\mathbf{h}^\tau$, the decision on UMSM can be made according to the solution for $\mathbf{x}^\mathbf{h}$, $\mathbf{y}^\mathbf{h}$, and $z^\mathbf{h}$ obtained by solving problem \eqref{prob:primal_in_h} with $\mathbf{h}=\mathbf{h}^\tau$.

However, problem \eqref{prob:primal_in_h} is still difficult to solve as the underlying distribution of $\mathbf{h}$ is unknown.
Hence, we resolve this difficulty by leveraging duality theory and a stochastic optimization method.
To this end, we first define a Lagrangian function $L$ as
\begin{equation}
\begin{aligned}
	&L(\mathbf{X}, \mathbf{Y}, \mathbf{z}, \bm{\lambdaup}, \bm{\muup}, \bm{\etaup}) \\
	&\quad= \mathbb{E}_\mathbf{h} \left[ \sum_{k\in\mathcal{K}} \left( w_k^u y_k^\mathbf{h} \cdot r_k(z^\mathbf{h}) + w_k^\varepsilon \cdot \varepsilon_k (z^\mathbf{h}) \right) \right] \\
	&\quad\quad + \sum_{k\in\mathcal{K}} \lambda_k \left( \mathbb{E}_\mathbf{h} \left[ x_{\ell(k)}^\mathbf{h} \cdot r_k(z^\mathbf{h}) \right] - \bar{B}_{\ell(k)} \right) \\
	&\quad\quad + \sum_{k\in\mathcal{K}} \mu_k \left( \mathbb{E}_\mathbf{h} \left[ y_k^\mathbf{h} \cdot r_k(z^\mathbf{h}) \right] - \bar{U}_k \right) \\
	&\quad\quad + \sum_{k\in\mathcal{K}} \eta_k \left( \mathbb{E}_\mathbf{h} \left[ \varepsilon_k(z^\mathbf{h}) \right] - \bar{E}_k \right)\\
	&\quad= \mathbb{E}_\mathbf{h} \left[ \sum_{k\in\mathcal{K}} \left(\lambda_k x_{\ell(k)}^\mathbf{h} + (w_k^u + \mu_k) y_k^\mathbf{h}\right) \cdot r_k(z^\mathbf{h}) \right.\\
	&\quad\qquad\qquad\qquad\qquad\quad \left. + \sum_{k\in\mathcal{K}} (w_k^\varepsilon + \eta_k) \cdot \varepsilon_k(z^\mathbf{h}) \right]\\
	&\quad\quad - \sum_{k\in\mathcal{K}} \left( \lambda_k \bar{B}_{\ell(k)} + \mu_k \bar{U}_k + \eta_k \bar{E}_k \right),
\end{aligned} \label{eq:lagrangian}
\end{equation}
where $\mathbf{X}=\{\mathbf{x}^\mathbf{h}\}_{\mathbf{h}}$, $\mathbf{Y}=\{\mathbf{y}^\mathbf{h}\}_{\mathbf{h}}$, and $\mathbf{z}=\{z^\mathbf{h}\}_{\mathbf{h}}$, $\bm{\lambdaup}=\{\lambda_k\}_{k\in\mathcal{K}}$, $\bm{\muup} = \{\mu_k\}_{k\in\mathcal{K}}$, $\bm{\etaup} = \{\eta_k\}_{k\in\mathcal{K}}$, and $\bm{\lambdaup}$, $\bm{\muup}$, and $\bm{\etaup}$ are the nonnegative Lagrangian multiplier vectors corresponding to the multicast constraints \eqref{eq:constraint_br2}, the unicast constraints \eqref{eq:constraint_uni2}, and the energy constraints \eqref{eq:constraint_E2}, respectively.
Using \eqref{eq:lagrangian}, the dual problem of problem \eqref{prob:primal_in_h} can be formulated as follows:
\begin{IEEEeqnarray}{c'l}
	\IEEEyesnumber\label{prob:dual}
	\minimize_{\bm{\lambdaup},\,\bm{\muup},\,\bm{\etaup}} & g(\bm{\lambdaup}, \bm{\muup}, \bm{\etaup}) \IEEEyessubnumber \\
	\subjto
	 & \bm{\lambdaup} \succeq \mathbf{0}_K, ~ \bm{\muup} \succeq \mathbf{0}_K, ~ \bm{\etaup} \succeq \mathbf{0}_K, \IEEEyessubnumber
\end{IEEEeqnarray}
where $\succeq$ denotes an elementwise inequality, $\mathbf{0}_K$ is a zero vector of size $K$, and
\begin{IEEEeqnarray}{rCc'l}
	\IEEEyesnumber\label{prob:dual_obj}
	g(\bm{\lambdaup}, \bm{\muup}, \bm{\etaup}) & = & \maximize_{\mathbf{X},\, \mathbf{Y},\, \mathbf{z}} & L(\mathbf{X}, \mathbf{Y}, \mathbf{z}, \bm{\lambdaup}, \bm{\muup}, \bm{\etaup}) \IEEEyessubnumber \\
	&&\subjto
	  & \eqref{eq:constraint_scheduling2}, \eqref{eq:xyz_region2}. \IEEEyessubnumber
\end{IEEEeqnarray}
It is noteworthy that although problem \eqref{prob:primal_in_h} is not a convex program, in our case, the duality gap between problem \eqref{prob:primal_in_h} and its dual problem \eqref{prob:dual} vanishes, resulting in no loss of optimality in dual transformation.

\begin{theorem} \label{thm:zero-duality-gap}
The strong duality holds between problem \eqref{prob:primal_in_h} and problem \eqref{prob:dual}.
\end{theorem}

\begin{proof}
As the proof of Theorem \ref{thm:zero-duality-gap} is analogous to that in our previous work, we omit it here for brevity.
Please refer to Theorem 6 in \cite{kim2022low} for the detailed proof.
\end{proof}

We now explain how to solve problem \eqref{prob:dual}.
We first focus on \eqref{eq:lagrangian}.
We can observe that the decision variables $\mathbf{X}$, $\mathbf{Y}$, and $\mathbf{z}$ in \eqref{prob:dual_obj} do not participate in the second term outside the expectation, and the first expectation term can be separated for each channel vector.
Hence, for any given $\bm{\lambdaup}$, $\bm{\muup}$, and $\bm{\etaup}$, solving the maximization in \eqref{prob:dual_obj} is equivalent to solving the following subproblems separately for every channel vector $\mathbf{h}$:
\begin{IEEEeqnarray}{c'l}
	\IEEEyesnumber\label{prob:dual_sep}
	\maximize_{\mathbf{x}^\mathbf{h},\,\mathbf{y}^\mathbf{h},\,z^\mathbf{h}} & \sum_{k\in\mathcal{K}} \left(\lambda_k x_{\ell(k)}^\mathbf{h} + (w_k^u + \mu_k) y_k^\mathbf{h}\right) \cdot r_k(z^\mathbf{h}) \IEEEeqnarraynumspace \IEEEyessubnumber \\
	 & + \sum_{k\in\mathcal{K}} (w_k^\varepsilon + \eta_k) \cdot \varepsilon_k(z^\mathbf{h}) \IEEEyessubnumber \\
	\subjto
	 & \sum_{c\in\mathcal{C}} x_c^\mathbf{h} + \sum_{k\in\mathcal{K}} y_k^\mathbf{h} \le 1, \IEEEyessubnumber \\
	 & \mathbf{x}^\mathbf{h} \in \mathcal{X}, ~ \mathbf{y}^\mathbf{h} \in \mathcal{Y}, ~ z^\mathbf{h} \in \mathcal{Z}. \IEEEyessubnumber
\end{IEEEeqnarray}
We will elaborate on how to solve problem \eqref{prob:dual_sep} in the upcoming section.
Therefore, in this section, we proceed with the assumption that a solution to problem \eqref{prob:dual_sep} can be obtained.

Let us return to solving problem \eqref{prob:dual}.
Although we can solve the maximization in problem \eqref{prob:dual_obj} by solving problem \eqref{prob:dual_sep} for every channel realization, problem \eqref{prob:dual} is still a stochastic programming problem that typically requires knowledge of the distribution to solve.
However, since problem \eqref{prob:dual} is convex stochastic programming with respect to $\{\bm{\lambdaup}, \bm{\muup}, \bm{\etaup}\}$, it can be solved using the well-known stochastic subgradient method.
The Lagrangian multipliers are iteratively updated as
\begin{align}
	\lambda_k^{\tau+1} &= \left[ \lambda_k^\tau - \zeta_\lambda^\tau \cdot \big( x_{\ell(k)}^\tau \cdot r_k(z^\tau) - \bar{B}_{\ell(k)} \big) \right]^+, ~ \forall k\in\mathcal{K}, \label{eq:update_lambda}\\
	\mu_k^{\tau+1} &= \left[ \mu_k^\tau - \zeta_\mu^\tau \cdot \big( y_k^\tau \cdot r_k(z^\tau) - \bar{U}_k \big) \right]^+, ~ \forall k\in\mathcal{K}, \label{eq:update_mu}\\
	\eta_k^{\tau+1} &= \left[ \eta_k^\tau - \zeta_\eta^\tau \cdot \big( \varepsilon_k(z^\tau) - \bar{E}_k \big) \right]^+, ~ \forall k\in\mathcal{K}, \label{eq:update_eta}
\end{align}
where $[\cdot]^+=\max\{0,\cdot\}$, $\lambda_k^\tau$, $\mu_k^\tau$, and $\eta_k^\tau$ are the Lagrangian multipliers in time slot $\tau$, $\zeta_\lambda^\tau$, $\zeta_\mu^\tau$, and $\zeta_\eta^\tau$ are their corresponding step sizes, and $\{x_c^\tau\}_{c\in\mathcal{C}}$, $\{y_k^\tau\}_{k\in\mathcal{K}}$, and $z^\tau$ are the solution to problem \eqref{prob:dual_sep} with $\mathbf{h}=\mathbf{h}^\tau$, $\lambda_k=\lambda_k^\tau$, $\mu_k=\mu_k^\tau$, and $\eta_k=\eta_k^\tau$, $\forall k\in\mathcal{K}$.
Note that the Lagrangian multipliers will converge to the optimal solution, $\{\bm{\lambdaup}^*, \bm{\muup}^*, \bm{\etaup}^*\}$, to problem \eqref{prob:dual} if the step sizes are set to be nonnegative and square summable, but not summable, i.e., for all $i\in\{\lambda, \mu, \eta\}$,
\begin{equation} \label{eq:stepsize_condition}
	\zeta_i^\tau \ge 0, ~ \sum_{\tau=1}^\infty \zeta_i^\tau = \infty, ~ \text{and} ~ \sum_{\tau=1}^\infty (\zeta_i^\tau)^2 < \infty.
\end{equation}
Algorithm \ref{alg:scheduling}, referred to as the UMSM scheduling algorithm, summarizes the process we have described thus far for solving the UMSM problem.

\begin{algorithm}[!t]
\small
\DontPrintSemicolon
Initialize: $\tau=1$, $\bm{\lambdaup}^\tau=\mathbf{0}_K$, $\bm{\muup}^\tau=\mathbf{0}_K$, and$\bm{\etaup}^\tau=\mathbf{0}_K$.\\
Determine $\eta_\lambda^\tau$, $\eta_\mu^\tau$, and $\eta_\eta^\tau$ such that \eqref{eq:stepsize_condition} is satisfied.\\
\For{each time slot $\tau$}{
	Obtain $\mathbf{h}^\tau$ by any channel estimation technique.\\
	Using Algorithm \ref{alg:selection}, solve problem \eqref{prob:dual_sep} with $\mathbf{h}=\mathbf{h}^\tau$, $\bm{\lambdaup}=\bm{\lambdaup}^\tau$, $\bm{\muup}=\bm{\muup}^\tau$, and $\bm{\etaup}=\bm{\etaup}^\tau$.\\
	Transmit the signal generated by the obtained solution.\\
	Calculate $\bm{\lambdaup}^{\tau+1}$, $\bm{\muup}^{\tau+1}$, and $\bm{\etaup}^{\tau+1}$ using \eqref{eq:update_lambda}, \eqref{eq:update_mu}, and \eqref{eq:update_eta}, respectively.\\
	$\tau \gets \tau+1$
}
\caption{UMSM Scheduling Algorithm}
\label{alg:scheduling}
\end{algorithm}

\section{Selection on Unicast/Multicast, IoT Sensor, and Modulation}
\label{sec:joint_selection}

In this section, we develop an algorithm called the UMSM selection algorithm to solve problem \eqref{prob:dual_sep}.
This algorithm involves determining whether to provide a multicast service for a specific task or a unicast service to a specific sensor, and selecting the modulation scheme, ARSK or ADSK, for a given channel realization $\mathbf{h}$ in a time slot.
For brevity, hereafter, we omit the superscript $\mathbf{h}$ and consider problem \eqref{prob:dual_sep} for a fixed channel in a generic time slot as
\begin{IEEEeqnarray}{c'l}
	\IEEEyesnumber\label{prob:dual_sep2}
	\maximize_{\mathbf{x},\,\mathbf{y},\,z} & \sum_{k\in\mathcal{K}} \left(\lambda_k x_{\ell(k)} + (w_k^u + \mu_k) y_k\right) \cdot r_k(z) \IEEEeqnarraynumspace \IEEEnonumber \\
	 & + \sum_{k\in\mathcal{K}} (w_k^\varepsilon + \eta_k) \cdot \varepsilon_k(z) \IEEEyessubnumber \label{eq:dual_sep2_obj}\\
	\subjto
	 & \sum_{c\in\mathcal{C}} x_c + \sum_{k\in\mathcal{K}} y_k \le 1, \IEEEyessubnumber \label{eq:dual_sep2_const} \\
	 & \mathbf{x} \in \mathcal{X}, ~ \mathbf{y} \in \mathcal{Y}, ~ z \in \mathcal{Z}. \IEEEyessubnumber
\end{IEEEeqnarray}
This problem is a nonconvex problem because of its integer variables.
Moreover, as can be inferred from \eqref{eq:E}, \eqref{eq:Ri}, \eqref{eq:Rktau}, and \eqref{eq:Ektau}, the throughput and harvested energy functions, $r_k$ and $\varepsilon_k$, in the objective function are neither concave nor convex.
To overcome these challenges, we develop a heuristic algorithm that can effectively solve problem \eqref{prob:dual_sep2} in an optimal way.

First, suppose that the modulation scheme is fixed to either ARSK or ADSK.
In this case, $z$ is fixed to either $1$ or $0$, and only the decision variables $\mathbf{x}$ and $\mathbf{y}$ remain in problem \eqref{prob:dual_sep2}.
Also, the second summation term in \eqref{eq:dual_sep2_obj} is independent of the decision variables.
As a result, we can simplify problem \eqref{prob:dual_sep2} by reducing it to the following formulation:
\begin{IEEEeqnarray}{c'l}
	\IEEEyesnumber\label{prob:xy}
	\maximize_{\mathbf{x},\,\mathbf{y}} & \sum_{k\in\mathcal{K}} \left(\lambda_k x_{\ell(k)} + (w_k^u + \mu_k) y_k\right) \cdot r_k(z) \IEEEeqnarraynumspace \IEEEyessubnumber \\
	\subjto
	 & \eqref{eq:dual_sep2_const}, ~ \mathbf{x} \in \mathcal{X}, ~ \mathbf{y} \in \mathcal{Y}. \IEEEyessubnumber
\end{IEEEeqnarray}
This is a natural consequence of the fact that the harvested energy of sensors is independent of scheduling when the modulation scheme is fixed.

Next, as can be seen from \eqref{eq:dual_sep2_const}, either a multicast service or a unicast service can be selected exclusively.
In the case of a multicast service, $y_k$ will be zero for all $k\in\mathcal{K}$, and an optimal task $c^*(z)$, whose content will be multicast to sensors in $\mathcal{K}(c^*(z))$, can be obtained as
\begin{equation}
	c^*(z) = \argmax_{c\in\mathcal{C}} ~ b(c;z) \triangleq \sum_{k\in\mathcal{K}(c)} \lambda_k \cdot r_k(z). \label{eq:c_star}
\end{equation}
Similarly, in the case of a unicast service, $x_c$ will be zero for all $c\in\mathcal{C}$, and an optimal sensor $k^*(z)$ to be selected as a unicast service receiver can be obtained as
\begin{equation}
	k^*(z) = \argmax_{k\in\mathcal{K}} ~ u(k;z) \triangleq (w_k^u+\mu_k) \cdot r_k(z). \label{eq:k_star}
\end{equation}
Lastly, by introducing a dummy indicator variable $\mathbf{1}_u(z)$ that takes a value of zero if $b(c^*(z);z)$ is greater than $u(k^*(z);z)$, and one otherwise, the optimal modulation scheme $z^*$ can be easily obtained as
\begin{equation}
\begin{aligned}[b]
	&z^* = \max_{z\in\{0,1\}} ~ \Bigg\{ (1-\mathbf{1}_u(z)) \cdot b(c^*(z);z) \\
	&\qquad\qquad + \mathbf{1}_u(z) \cdot u(k^*;z) + \sum_{k\in\mathcal{K}} (w_k^\varepsilon + \eta_k) \cdot \varepsilon_k(z) \Bigg\}
\end{aligned}.\label{eq:z_star}
\end{equation}
With $z^*$ obtained by \eqref{eq:z_star}, $\mathbf{x}^*=\{x_c^*\}_{c\in\mathcal{C}}$ and $\mathbf{y}^*=\{y_k^*\}_{k\in\mathcal{K}}$ can be finally determined as follows:
$\mathbf{x}^*$ such that $x_{c^*(z^*)}^* = 1 - \mathbf{1}_u(z^*)$ and $x_c^* = 0$ for all $c\in\mathcal{C}\setminus\{c^*(z^*)\}$ and $\mathbf{y}^*$ such that $y_{k^*(z^*)}^* = \mathbf{1}_u(z^*)$ and $y_k^* = 0$ for all $k\in\mathcal{K}\setminus\{k^*(z^*)\}$.
Algorithm \ref{alg:selection}, referred to as the UMSM selection algorithm, summarizes the process described in this section.


\begin{algorithm}[!t]
\small
\DontPrintSemicolon
Initialize: $\mathbf{x}^* \gets \mathbf{0}_C$ and $\mathbf{y}^* \gets \mathbf{0}_K$.\\
Obtain $c^*(z)$ and $k^*(z)$ using \eqref{eq:c_star} and \eqref{eq:k_star}, respectively.\\
\leIf{$b(c^*(z);z) > u(k^*(z);z)$}{$\mathbf{1}_u(z) \gets 0$}{$\mathbf{1}_u(z) \gets 1$.}
Obtain $z^*$ using \eqref{eq:z_star}.\\
$x_{c^*(z^*)}^* \gets 1-\mathbf{1}_u(z^*)$ and $y_{k^*(z^*)}^* \gets \mathbf{1}_u(z^*)$.\\
\Return{$\{\mathbf{x}^*, \mathbf{y}^*, z^*\}$}
\caption{UMSM Selection Algorithm}
\label{alg:selection}
\end{algorithm}

\section{Simulation Results}
\label{sec:sim}

In this section, we present the results of our simulations to assess the effectiveness of our proposed algorithms.
In the following simulations, for the receiver model, we set the rectifier noise variances in Fig. \ref{fig:Receiver architecture} to $\sigma_p^2 = -40$\,dBm and $\sigma_n^2 = -40$\,dBm.
In Fig. \ref{fig:Receiver circuit}, we use Skyworks SMS7630 Schottky diodes with $I_s=5$\,$\mu$A, $\rho=1.05$, and $V_{th}=25.85$\,mV.
To maintain perfect impedance matching with $50$\,$\Omega$ impedance commonly used in RF applications, we set $R_r=50$\,$\Omega$.
Additionally, we set the capacitances of the two capacitors to $100$\,pF and the load resistance $R_L$ to $150$\,$\Omega$.
For the signal model, we assume that the transmitted signal is composed of two tones with frequencies of $f_1 = 900$\,MHz and $f_2 = 1.8$\,GHz, respectively, as in \cite{kim2019dual}.
We also assume a $4$-ary transmission, i.e., $M=4$, as shown in Figs. \ref{fig:Throughput_curves} and \ref{fig:Energy_curves}.
The transmit power budget, $P_t$, of the H-AP is set to $40$\,dBm.
For the channel model, we adopt the empirical path loss model measured in a room-to-room non-line-of-sight propagation condition \cite{xu2007indoor}, where the path loss in dB over distance $d_m$ in meters is set to $29.3+60.5\log_{10}d_m$.
We set the antenna gains of the H-AP and the sensors to $3$\,dBi and $0$\,dBi, respectively.
Additionally, we consider the shadow fading with a standard deviation of $1.8$\,dB and the Rayleigh small-scale fading with unit variance.
Within Algorithm \ref{alg:scheduling}, to ensure convergence of the algorithm, the step sizes, $\zeta_\lambda$ in \eqref{eq:update_lambda}, $\zeta_\mu$ in \eqref{eq:update_mu}, and $\zeta_\eta$ in \eqref{eq:update_eta}, are all set to $10^3/(10^3+\tau-1)$, which satisfy the conditions described in \eqref{eq:stepsize_condition}.
Also, the throughput and the harvested energy metrics are measured in units of bits per channel use and $\mu$Wh, respectively.

In the first analysis, we focus on the simulation results obtained as the energy weights of sensors increase, while the throughput weights remain fixed at $1$.
To this end, we consider a scenario where $7$ sensors are performing $3$ tasks, i.e., $C=3$ and $K=7$.
Specifically, sensors $1$ and $2$ perform task $1$, sensors $3$, $4$, and $5$ perform task $2$, and sensors $6$ and $7$ perform task $3$.
These correspond to $\mathcal{K}(1)=\{1,2\}$, $\mathcal{K}(2)=\{3,4,5\}$, $\mathcal{K}(3)=\{6,7\}$, $\ell(1)=\ell(2)=1$, $\ell(3)=\ell(4)=\ell(5)=2$, and $\ell(6)=\ell(7)=3$.
Additionally, we assume that sensors $1$ to $7$ are located sequentially at distances of $4$\,m, $5$\,m, $4$\,m, $5$\,m, $6$\,m, $4$\,m, $5$\,m away from the H-AP, respectively.
In this scenario, we execute Algorithm \ref{alg:scheduling} for $100,000$ time slots with the following parameter settings for the minimum requirements: $\bar{B}_c=0.05$, $\forall c$, and $\bar{U}_k=0.2$, $\forall k$.
On the other hand, to analyze the simulation results more clearly with respect to the changes in the energy weights, we set $\bar{E}_k$, $\forall k$, to zero.
Later, we will analyze the simulation results for non-zero values of $\bar{E}_k$, $\forall k$.

\begin{figure}[!t]
\centerline{\includegraphics[width=.6\columnwidth]{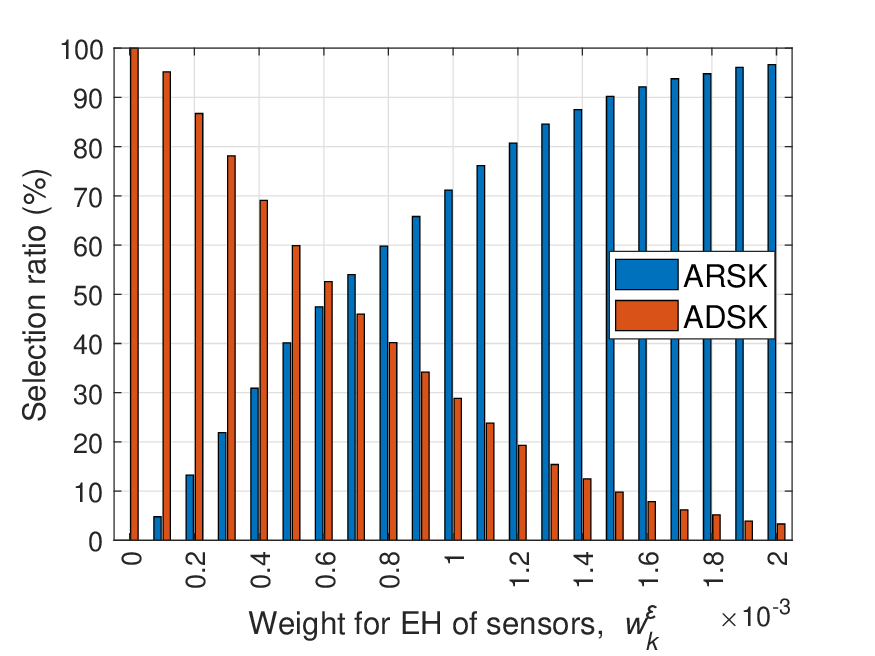}}
\caption{Selection ratios of modulation schemes with respect to the weights for EH of sensors.}
\label{fig:selection rate_we}
\end{figure}

\begin{figure}[!t]
\centerline{\includegraphics[width=.6\columnwidth]{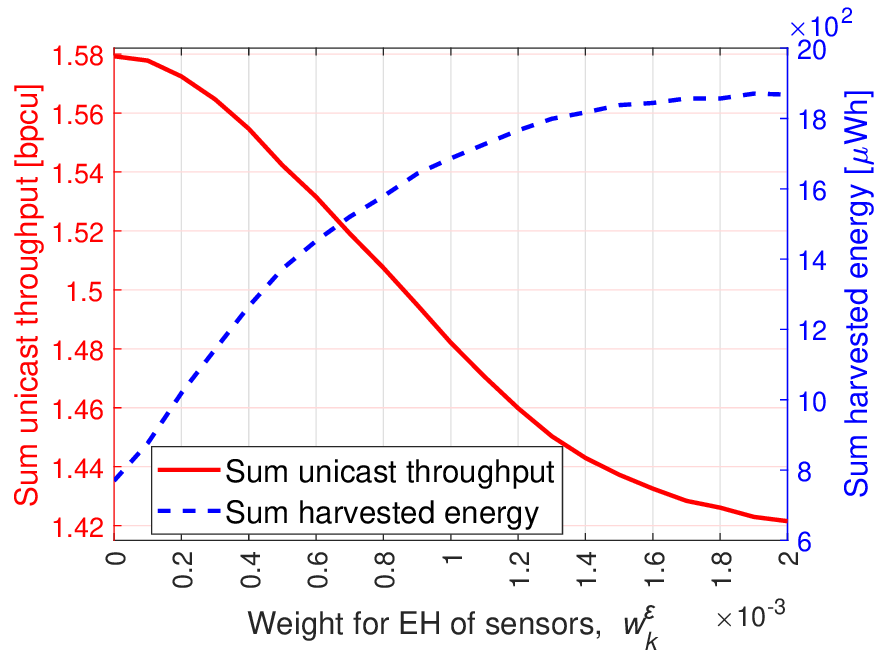}}
\caption{Performance with respect to the weights for EH of sensors.}
\label{fig:performance_we}
\end{figure}

Figs. \ref{fig:selection rate_we} and \ref{fig:performance_we} show the variations in selection ratios, as well as the performance metrics, sum unicast throughput and sum harvested energy, for the ARSK and ADSK modulation schemes, respectively, as the energy weights of sensors increase.
It is worth noting that, as depicted in Figs. \ref{fig:Throughput_curves} and \ref{fig:Energy_curves}, ADSK outperforms ARSK in terms of throughput performance, whereas ARSK outperforms ADSK in terms of harvested energy performance.
Hence, the results shown in Figs. \ref{fig:selection rate_we} and \ref{fig:performance_we} are reasonable since the selection ratio for ARSK increases with the increase in the energy weights of sensors, leading to a monotonic rise in sum harvested energy and a monotonic decline in sum unicast throughput.


\begin{figure}[!t]
\centerline{\includegraphics[width=.7\columnwidth]{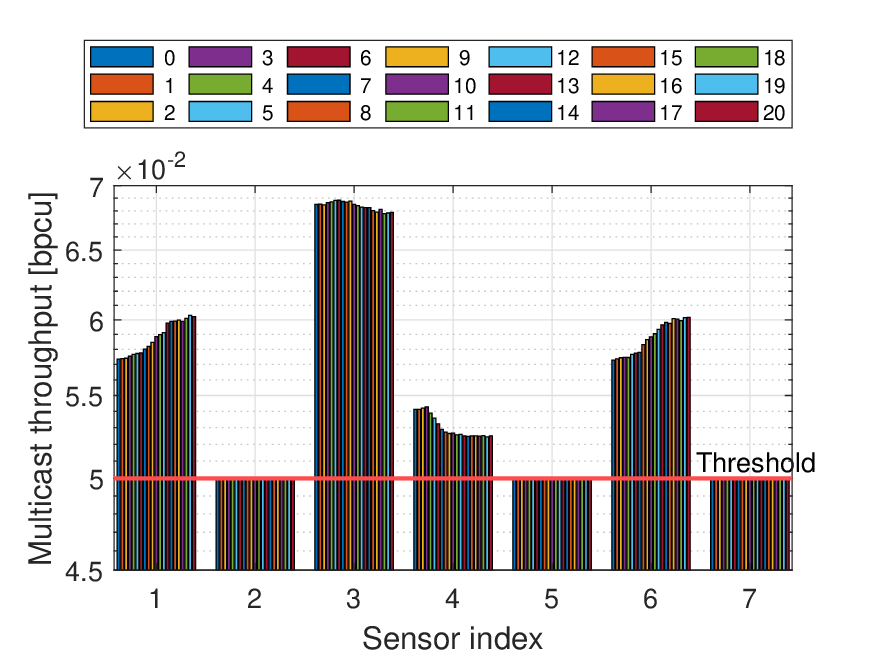}}
\caption{Multicast throughput for each sensor.}
\label{fig:constraint_multicast_we}
\end{figure}

\begin{figure}[!t]
\centerline{\includegraphics[width=.7\columnwidth]{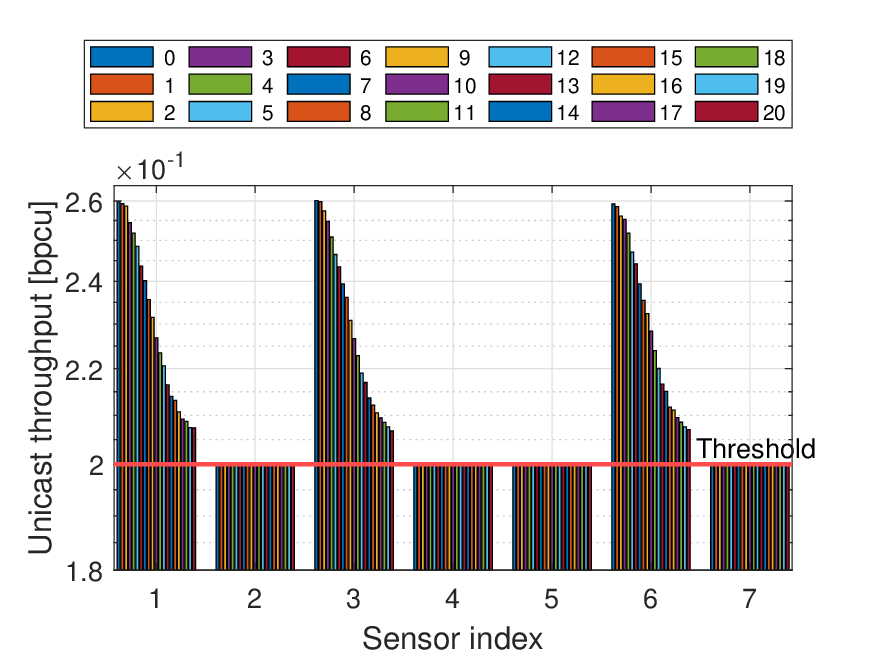}}
\caption{Unicast throughput for each sensor.}
\label{fig:constraint_unicast_we}
\end{figure}

Fig. \ref{fig:constraint_multicast_we} shows the multicast throughput of each sensor, with $21$ bars corresponding to each sensor index.
For any $l\in\{0,1,\ldots,20\}$, the $l$th bar for each sensor index represents its multicast throughput when the energy weights are set as to be as $w_k^\varepsilon = l \times 10^{-4}$, $\forall k$.
For example, the $21$ bars for sensor $3$ indicate its multicast throughputs for $21$ different energy weights from $0$ to $0.002$.
The threshold line represents the minimum average multicast throughput required, set at $0.05$.
Recall that when multicast service for task $c$ is scheduled, all sensors in $\mathcal{K}(c)$ receive the multicast service simultaneously from the H-AP.
For example, sensors $1$ and $2$, sensors $3$, $4$, and $5$, and sensors $6$ and $7$ always receive multicast services together.
However, due to their varying distances from the H-AP, the multicast throughputs of these sensor pairs can differ.
Notably, the multicast throughputs of sensors $2$, $5$, and $7$ \textemdash{} the ones farthest from the H-AP for each task \textemdash{} barely reach the threshold value of $0.05$ with negligible errors and do not exceed it significantly.
This is because scheduling as few multicast services as possible is desirable, as long as the constraints in \eqref{eq:constraint_br} are met, since the multicast throughput does not affect the objective value.

Fig. \ref{fig:constraint_unicast_we} shows the unicast throughput for each sensor, using the same format as in Fig. \ref{fig:constraint_multicast_we}, with a threshold of $0.2$ since $\bar{U}_k=0.2$, $\forall k$.
Unlike multicast throughput, unicast throughput is a performance metric for each sensor, and it directly affects the objective value.
Therefore, it is obvious that a higher sum unicast throughput can be achieved by selecting a sensor with the highest instantaneous channel gain for each time slot and providing a unicast service thereto.
The figure demonstrates that our algorithm prioritizes providing unicast service to sensors with the highest instantaneous channel gains, serving sensors $1$, $3$, and $6$ as much as possible since they are closest to the H-AP, while serving the remaining sensors only to the extent that the unicast throughput constraints in \eqref{eq:constraint_uni} are met.
An interesting observation is that since sensors $1$, $3$, and $6$ have the same distance from the H-AP and equal throughput weights, they have nearly the same unicast throughput distribution over $21$ bars, despite the randomness resulting from the stochastically generated channel gains over $100,000$ time steps.
Moreover, for these sensors, as the bars shift towards the right, their energy weights increase, leading to a decrease in unicast throughput since ARSK is exploited more, as already shown in Fig. \ref{fig:selection rate_we}.

Next, we focus on the impact of the minimum average harvested energy requirement on the system.
To clearly illustrate its effects, in the following simulation, we set the energy weights and the minimum average harvested energy requirements to zero and $3.6$\,$\mu$Wh, respectively, for all sensors, i.e., $w_k^\varepsilon = 0$ and $\bar{E}_k = 3.6$\,$\mu$Wh, $\forall k$.
%
Additionally, unlike the previous simulation where the throughput requirements of all sensors are equal, in this simulation we set different throughput requirements for each sensor as follows: $\bar{B}_1=0.04$, $\bar{B}_2=0.02$, and $\bar{B}_3=0.05$.
This means that the minimum average multicast throughputs for the three tasks are set differently.
Furthermore, each sensor has its own minimum average unicast throughput requirement as follows: $\bar{U}_1=\bar{U}_4=0.2$, $\bar{U}_2=\bar{U}_5=\bar{U}_7=0.1$, and $\bar{U}_3=\bar{U}_6=0.3$.

\begin{figure}[!t]
\centerline{\includegraphics[width=.7\columnwidth]{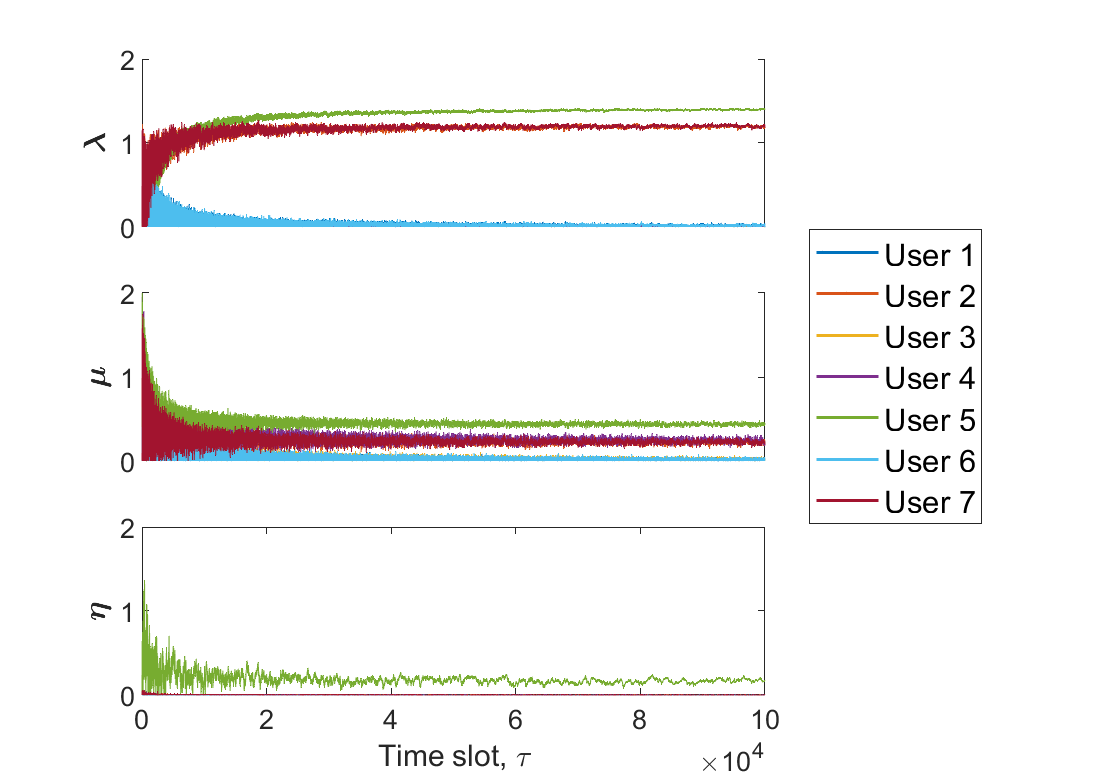}}
\caption{Convergence of the Lagrangian multipliers.}
\label{fig:Lagrangian_multiplier}
\end{figure}

We first demonstrate the convergence of our proposed UMSM scheduling algorithm in Fig. \ref{fig:Lagrangian_multiplier}, which comprises three graphs showing the changes in different parameters over time slots for all sensors.
The top graph depicts the changes in $\bm{\lambdaup}=\{\lambda_1, \ldots, \lambda_7\}$, the Lagrangian multiplier for the multicast throughput requirements.
The middle graph shows the changes in $\bm{\muup}=\{\mu_1, \ldots, \mu_7\}$, the Lagrangian multiplier for the unicast throughput requirements.
The bottom graph displays the changes in $\bm{\etaup}=\{\eta_1, \ldots, \eta_7\}$, the Lagrangian multiplier for the harvested energy requirements.
From these graphs, we can see that each Lagrangian multiplier eventually converges to a certain stationary point, even though the stationary points can be different due to the differences in sensors in terms of the minimum required multicast and unicast throughputs, the minimum required harvested energy, and the channel conditions from the H-AP.

\begin{figure}[!t]
\centerline{\includegraphics[width=.6\columnwidth]{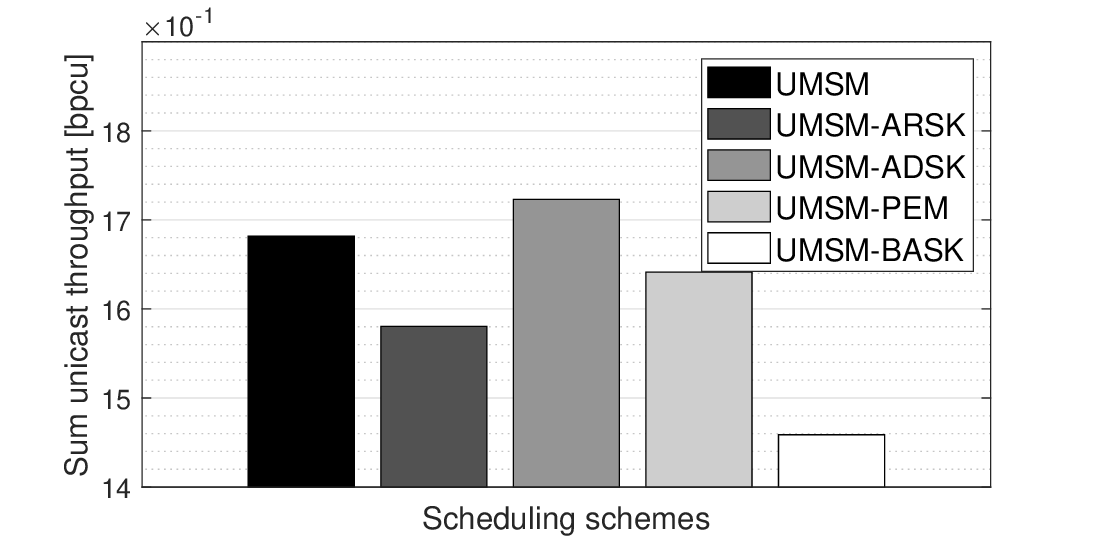}}
\caption{Sum unicast throughput for different scheduling schemes.}
\label{fig:sum_unicast_throughput}
\end{figure}

\begin{figure}[!t]
\centerline{\includegraphics[width=.6\columnwidth]{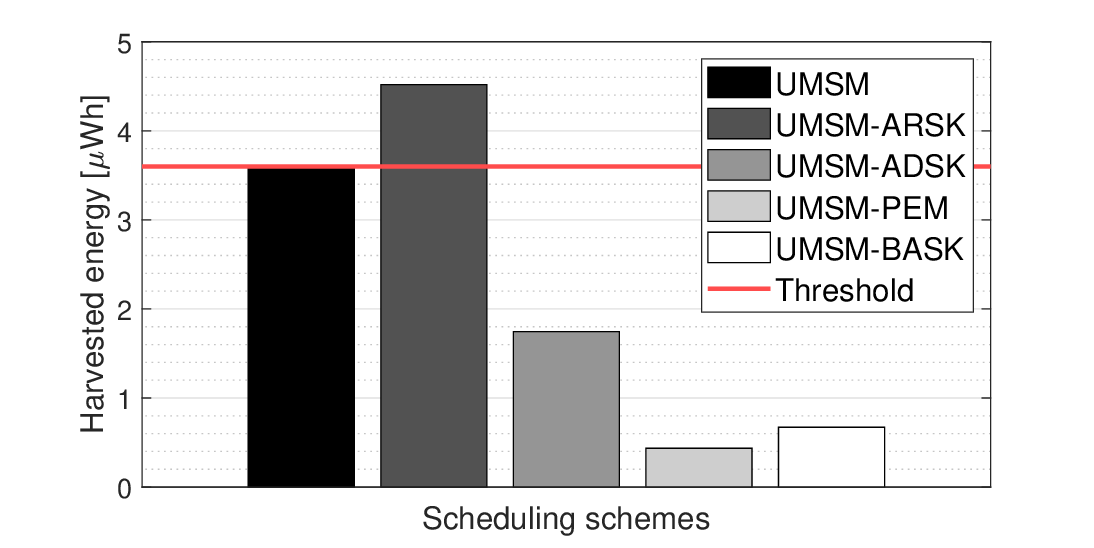}}
\caption{Harvested energy of sensor $5$, the farthest from the H-AP, for different scheduling schemes.}
\label{fig:harvested_energy_of_sensor5}
\end{figure}

We present a comparison of the performance of our proposed UMSM scheduling algorithm with four baseline algorithms, UMSM-ARSK, UMSM-ADSK, UMSM-PEM, and UMSM-BASK, in terms of the sum unicast throughput in Fig. \ref{fig:sum_unicast_throughput}, as well as the harvested energy of sensor $5$ under the worst channel condition in Fig. \ref{fig:harvested_energy_of_sensor5}.
The baseline scheduling algorithms also attempt to solve the UMSM scheduling problem based on Algorithm \ref{alg:scheduling} with Algorithm \ref{alg:selection}, similar to our proposed UMSM scheduling algorithm.
However, while our UMSM scheduling algorithm opportunistically chooses an appropriate modulation scheme for each time slot, UMSM-ARSK, UMSM-ADSK, UMSM-PEM, and UMSM-BASK use fixed modulation schemes of ARSK, ADSK, PEM, and BASK, respectively.
First of all, we can observe that the UMSM-PEM and UMSM-BASK scheduling algorithms provide very low harvested energy, and they cannot even satisfy the minimum required harvested energy constraint of sensor $5$, which is under the worst channel condition due to its location farthest from the H-AP.
This is because, as previously demonstrated in \cite{kim2019dual}, PEM and BASK have much lower EH performance compared to ARSK and ADSK.
Meanwhile, as shown in Fig. \ref{fig:Throughput_curves}, ADSK outperforms ARSK in terms of throughput, and thus, the UMSM-ADSK scheduling algorithm achieves a higher sum unicast throughput than the other scheduling algorithms, including our UMSM scheduling algorithm.
However, due to the lower performance of ADSK on harvested energy compared to ARSK, the UMSM-ADSK scheduling algorithm cannot satisfy the minimum required harvested energy constraint of sensor $5$, like UMSM-PEM and UMSM-BASK.
Conversely, the UMSM-ARSK scheduling algorithm achieves higher harvested energy than the other scheduling algorithms since ARSK has better performance on harvested energy, but it has a lower sum unicast throughput compared to our UMSM scheduling algorithm.
It is important to note that since $w_k^\varepsilon$ is set to zero in this simulation, achieving higher harvested energy does not increase the objective value.
Thus, achieving harvested energy up to the threshold of $3.6$\,$\mu$Wh is sufficient, and achieving the highest possible sum unicast throughput will increase the objective value.
As a result, we can conclude that our proposed UMSM scheduling algorithm outperforms the other scheduling algorithms while ensuring the given constraints.

\begin{figure}[!t]
\centerline{\includegraphics[width=.6\columnwidth]{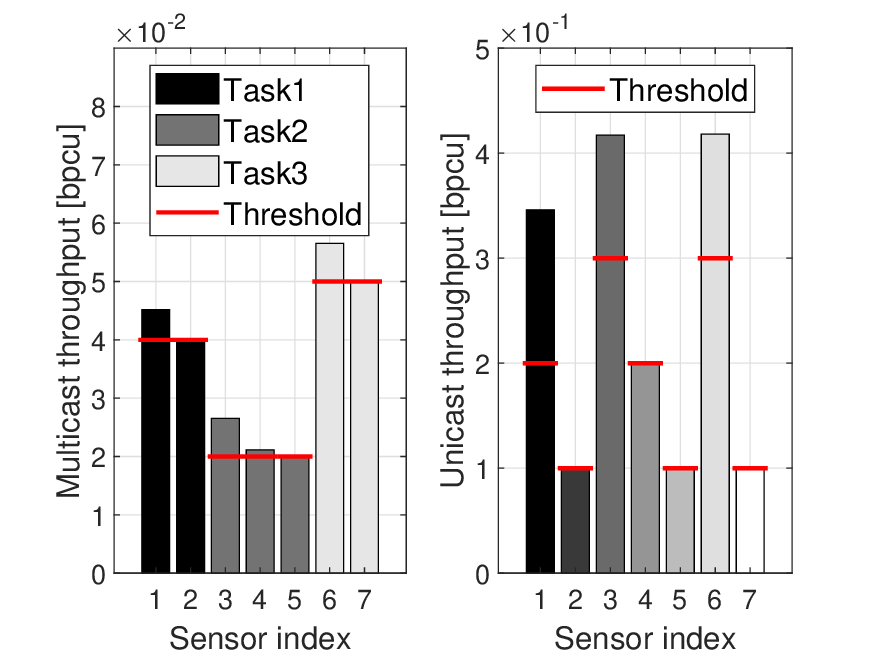}}
\caption{Multicast throughput for each sensor.}
\label{fig:throughput_each_sensor}
\end{figure}

Lastly, in Fig. \ref{fig:throughput_each_sensor}, we present the multicast and unicast throughputs of each sensor achieved by our UMSM scheduling algorithm.
The left graph in the figure shows the multicast throughput of each sensor with the minimum required multicast throughput for each task, while the right graph shows the unicast throughput of each sensor with the minimum required unicast throughput for each sensor.
In the left graph, we observe that multicast service is provided only up to the level at which the multicast throughput constraints for tasks are satisfied.
This is because the multicast throughput does not participate the objective value.
On the other hand, in the right graph, our UMSM scheduling algorithm satisfies the minimum required unicast throughputs of sensors and then intensively schedules sensors $1$, $3$, and $6$, which are closest to the H-AP and likely to have high channel gains.
This is because scheduling sensors with better channel gains results in a higher sum unicast throughput, which is the objective value itself since the energy weights are zero.
Therefore, our UMSM scheduling algorithm not only satisfies all the minimum required multicast and unicast throughputs constraints, as well as the minimum required harvested energy constraints, but also performs well by appropriately choosing between multicast and unicast services, determining the appropriate sensor to serve if unicast, and selecting the optimal modulation scheme between ARSK and ADSK.

%

%
%
%

\section{Conclusion}
\label{sec:conc}

In this paper, we have proposed a novel constraints-aware scheduling algorithm, UMSM, for SWIPT-based IoT sensor networks with multiple tasks and throughput and energy harvesting constraints.
Our proposed algorithm jointly optimizes the use of both multicast and unicast communications, as well as the selection of modulation schemes to maximize the weighted sum of the unicast throughput and harvested energy of all sensors while satisfying the minimum required multicast and unicast throughputs and harvested energy.
Through simulations, we have demonstrated that our UMSM algorithm outperforms two other scheduling algorithms with fixed modulation schemes, UMSM-ARSK and UMSM-ADSK.
Furthermore, we have shown that our UMSM algorithm well satisfies all minimum required multicast and unicast throughputs as well as the minimum required harvested energy constraints by opportunistically exploiting channel stochasticity while still achieving a high sum unicast throughput.
Overall, our proposed UMSM algorithm provides an efficient and effective solution for energy-constrained wireless sensor networks with multiple tasks and constraints.
Our work makes a valuable contribution to the field of SWIPT-based wireless sensor networks and opens up avenues for further research on energy-efficient solutions for larger-scale SWIPT-based IoT sensor networks.

\bibliographystyle{IEEEtran}
\bibliography{IEEEabrv,mybib}

\end{document}